\documentclass[10pt,letterpaper,oneside]{amsart}
\usepackage[english]{babel}
\usepackage{amssymb}
\usepackage{geometry,mathrsfs,hyperref,xcolor}
\usepackage{algorithm}
\usepackage{algpseudocode}
\usepackage{mathtools}
\usepackage{subcaption}
\usepackage{comment}
\usepackage[backend=biber]{biblatex}
\usepackage{todonotes}

\newcommand{\TeXmacs}{T\kern-.1667em\lower.5ex\hbox{E}\kern-.125emX\kern-.1em\lower.5ex\hbox{\textsc{m\kern-.05ema\kern-.125emc\kern-.05ems}}}
\newcommand{\assign}{:=}

\newcommand{\tmop}[1]{\ensuremath{\operatorname{#1}}}

\theoremstyle{plain}
\newtheorem{theorem}{Theorem}[section]
\newtheorem{corollary}[theorem]{Corollary}

\newtheorem*{theorem*}{Theorem}
\newtheorem*{lemma*}{Lemma}
\newtheorem*{prop*}{Proposition}
\newtheorem*{cor*}{Corollary}
\newtheorem*{conj*}{Conjecture}

\theoremstyle{definition}
\newtheorem*{definition*}{Definition}
\newtheorem{definition}[theorem]{Definition}

\theoremstyle{remark}
\newtheorem*{rem*}{Remark}

\newtheorem*{example*}{Example}

\numberwithin{theorem}{section}
\numberwithin{equation}{section} 
\numberwithin{figure}{section} 

\newcommand{\R}{\ensuremath{\mathbb{R}}}

\newcommand{\e}{\mathrm{e}}

\begin{document}

\title[Data-Driven Estimation of Real Cohomology]{Noncommutative Model Selection and the Data-Driven
Estimation of Real Cohomology Groups}

\thanks{$^{\ast}$Corresponding author}
\thanks{Araceli Guzmán-Tristán was supported by the CONAHCYT program "Estancias Posdoctorales por México para la Formación y Consolidación de las y los Investigadores por México". Antonio Rieser was supported by the US National
Science Foundation under grants No. DMS-1929284 and DMS-1928930, the first while in residence at the Institute for Computational and Experimental Research in Mathematics in Providence, RI, during the "Math + Neuroscience: Strengthening the Interplay Between Theory and Mathematics" program, and the second while in residence
at the Simons-Laufer Mathematical Sciences Research Institute in the spring of 2024 and in a program supported by the Mathematical Sciences Research Institute in the summer of 2022,
held in partnership with the the Universidad
Nacional Aut{\'o}noma de M{\'e}xico. Antonio Rieser was also supported by the
CONAHCYT Investigadoras y Investigadores por M{\'e}xico Project \#1076 and by
the grant N62909-19-1-2134 from the US Office of Naval Research Global and the
Southern Office of Aerospace Research and Development of the US Air Force
Office of Scientific Research. Eduardo Velázquez-Richards was supported by the UNAM Posdoctoral Program (POSDOC) while contributing to this work.}

\author[A. Guzm\'an-Trist\'an]{Araceli Guzm{\'a}n-Trist{\'a}n}

\author[A. Rieser]{Antonio Rieser$^{\ast}$}
\address{CIMAT Guanajuato, Calle Jalisco S/N, Colonia Valenciana,
Guanajuato, GTO, C.P. 36023, México}
\email{araceli.guzman@cimat.mx,antonio.rieser@cimat.mx}
\author[E. Vel\'azquez-Richards]{Eduardo Vel{\'a}zquez-Richards}
\address{IIMAS, Universidad Nacional Autónoma de México, Ciudad Universitaria, Coyoacán, Ciudad de México, C.P. 04510, México}
\email{eduardo.velazquez@iimas.unam.mx}
\begin{abstract}
  We propose three completely data-driven methods for estimating the
  real cohomology groups $H^k (X ; \R)$ of a compact metric-measure
  space $(X, d_X, \mu_X)$ embedded in a metric-measure space $(Y,d_Y,\mu_Y)$, given a
  finite set of points $S$ sampled from a uniform distrbution $\mu_X$ on
  $X$, possibly corrupted with noise from $Y$. We present the results of several computational experiments in the case that $X$
  is embedded in $\R^n$, where two of the three algorithms performed well.
\end{abstract}

{\maketitle}
\section{Introduction}
The problem of estimating the topological invariants of a compact metric-measure space from a finite set of sample points is a central question, and arguably
\emph{the} central question, of topological data analysis, although an
algorithmic solution to this problem has remained elusive since the
beginnings of the field.\footnote{One can naturally compute persistent
  (\v Cech or Vietoris-Rips) homology of a set of sample points - at least in regimes in which
  the computation is tractable - which, when the set is dense enough, can be shown to approximate the persistent homology of an underlying metric-measure space \cite{Chazal_etal_2016} in the bottleneck metric. However, since the persistent homology of
  a metric-measure space is not invariant under continuous
  transformations of the space, it is not a topological invariant, and
  automatic methods of extracting the (singular) homology of a metric-measure space from the persistent homology of point samples are
  still lacking.}  Nonetheless, while methods for estimating
the topological invariants of a space from samples have not yet been
found, a number of classical results do exist which illustrate what such
an algorithm might try to look for. For instance, it was shown in
\cites{Niyogi_Smale_Weinberger_2008,Chazal_etal_2009} that if $X$
is a closed manifold or a compact subset of $\R^n$, then there exist
certain ranges of $r>0$ such that, for a dense enough set of sample
points from $X$, the union of disks of radius $r$ with centers at the
sample points has the same homotopy type as $X$. More recently, it was
shown in \cites{Kalisnik_Lesnik_arXiv_2020} that similar results also
hold when one replaces the balls with appropriately constructed
ellipsoids, and \cite{Latschev_2001} proved a similar theorem where
the Vietoris-Rips complex of a metric space at scale $r>0$ takes the
place of the union of balls.

Despite the promise of these results, however, none of them reveal how
to choose the relevant free parameters for the model, i.e. the radii of the balls, the axes of the ellipses, or the scale for 
the Vietoris-Rips complex, starting
with a given a set of points, and, to the best of our knowledge, no
such proposal has been made to date. Historically, the question of how
to estimate the free parameters in these problems was ultimately put aside in favor of considering all
possible values of the parameters at once to construct a conglomerate invariant
which is now known as a persistence module, the basic object
of study in persistent homology \cites{Edelsbrunner_etal_2002,
	Carlsson_Zomorodian_2005, Carlsson_2009, Chazal_etal_2016}. Closely related to persistent homology are Euler characteristic curves 
	\cites{Smith_Zavala_2021,Laky_Zavala_2024, Amezquita_etal_2021}, in which the Euler characteristic of a parametrized filtration of a space is recorded 
	for every value of the parameter, which have the advantage of being more easily computable than persistent homology. Taken together, these two objects
  are currently the most
popular tools in topological data analysis, and in the past
two decades, they have led to a number of intriguing observations,
particularly in problems which benefit from describing one-dimensional
cycles inside small to medium-sized data sets \cites{DONUT}. Nonetheless, neither persistent homology nor Euler characteristic curves give direct estimates of the topological invariants of the space from which a set of points may have been sampled. Additionally, although there exist several commonly used heuristics for making hypotheses about the singular homology of a space given the output of these algorithms, methods for reliably and automatically extracting topological invariants from these objects have yet to be developed.

In this paper, we return to the original question of how to model a metric-measure space $(X,d_X,\mu_X)$ using a finite set of sample points
in order to estimate a given topological invariant of $(X,d)$. Our models, however,
will be analytic objects - in particular, semigroups of operators -
instead of topological spaces, and our goal will
be only to select a model which enables us to estimate a specific topological invariant of interest instead of trying to recover the space itself up to
homotopy. Additionally, our techniques are ultimately unrelated to those in the current topological data anlaysis literature, and are instead inspired by constructions in spectral and geometric data analysis such as Laplacian Eigenmaps \cite{Belkin_Niyogi_2003} and Diffusion Maps \cite{Coifman_Lafon_2006}. 

Specifically, we propose three related algorithms which use the spectral properties of certain operator semigroups to estimate the real cohomology groups of a metric-measure space
$(X,d_X,\mu_X)$  from point
samples, possibly corrupted with noise, where $X$ has been embedded
in an ambient metric space $(Y,d_Y)$ and the measure $\mu_X$ is uniform. Each of these algorithms works by translating the topological problem of estimating the $q$-dimensional cohomology of a space into the analytic problem of estimating the semigroup of operators generated by the Hodge-Laplacians of the space.
In order to do this, we first use the sample points to construct a family of candidate semigroups, after which we use the structure of the semigroups and several ideas inspired by statistical
model selection techniques in order to choose a model for the invariant of interest which is then used to make the estimation. 

The selection criteria which we have developed are
all designed so that the chosen semigroup roughly captures the ``most
geometry'' in a given dimension. The first method is based on relative von Neumann entropy, an operator-theoretic version of the relative Shannon entropy used in quantum information theory \cite{Wilde_2013}, which the first two authors also
used in \cite{Guzman_Rieser_2024} for clustering and dimension
reduction. We compare this to two new techniques, one using the trace, and the other using a natural metric on the space of Hilbert-Schmidt operators. These criteria are described further in Section \ref{sec:AlgRes}. Once a specific semigroup is chosen, we appeal to standard results in combinatorial Hodge theory in order to identify the real cohomology group
$H^q(X;\R)$ in dimension $q\geq 0$ with a
certain eigenspace of the operators in the semigroup, and the calculation of the Betti numbers follows. Finally, we note that while the 
general ideas of statistical model selection have guided the development
of our algorithms, 
traditional statistical methods don't apply directly to the analytic
problem which we have introduced here, as our problem is inherently operator-theoretic in nature.

\subsection{Related Work}

The spectral properties of graph Laplacians \cite{Chung_1997} and Hodge-Laplacians on simplicial complexes 
\cite{Horak_Jost_2013,Friedman_1998} have long been studied in combinatorics and algebraic topology \cite{Eckmann_1945}, largely in parallel to the 
corresponding objects in differential geometry \cite{Chavel_1984,Berline_etal_2004,}. More recent work \cite{Horak_Jost_2013} 
incorporates weights into combinatorial Hodge-Laplacians, a technique which we
also use here. In \cite{Bartholdi_etal_2012, Smale_Smale_2012}, the authors study a Hodge theory defined at different scales on metric spaces by constructing a variant of the Alexander-Spanier cohomology which incorporates the scale into its construction.

Starting with work on spectral clustering \cite{Luo_Wilson_Hancock_2003}, and continuing with Laplacian Eigenmaps \cites{Belkin_Niyogi_2003,Belkin_Niyogi_2006,Belkin_Niyogi_2008}, Diffusion Maps \cites{Coifman_Lafon_2006}, and Vector Diffusion Maps \cites{Singer_Wu_2012}, the spectral information in different graph Laplacians has been used for a variety of tasks in data analysis, in particular for data clustering and dimension reduction. These methods have enjoyed significant success, although in each of them there is a free parameter left for the user to choose, and which is typically chosen by hand. The first successful heuristic for choosing the free parameter in Diffusion Maps was reported in \cite{Shan_Daubechies_2022}, but it does not apply to the geometric model studied here, as the free parameter in our method depends only on the geometry of the data set and is a priori decoupled from the semigroup parameter. 

The most closely related work to the present article are the articles \cite{Rieser_FODS_2021, Guzman_Rieser_2024}. In the first of these, the second author introduced a method based on (classical) relative Shannon entropy to choose a free parameter to estimate the ``heat flow" on a data set, and then used this to develop a new data clustering algorithm. More recently, in \cite{Guzman_Rieser_2024}, the first and second authors proposed a method to estimate the heat semigroup on a metric-measure space by an empirical heat semigroup built from a graph, and then used that estimation as a step in both a topological method for data clustering and a dimension reduction method similar to Diffusion Maps. The present article may be considered to be a generalization of the methods in \cite{Guzman_Rieser_2024}, and particularly the clustering algorithm, to higher dimensional Hodge Laplacians. Indeed, the clustering algorithm from \cite{Guzman_Rieser_2024} is exactly the $0$-dimensional version of the von Neumann entropy method below, but with an additional step added to identify which points belong to each cluster. Curiously, neither
of the other two methods considered here worked well in the $0$-dimensional case on the examples we tested, although the reasons for that remain unclear.

\subsection{Contributions}
We introduce three new methods to produce a fully data-driven estimate of the Betti numbers of a manifold from uniformly distributed data, and we show that they work well on a number of synthetic data sets, as well as real data sets which can be presumed to be close to uniformly distributed. To the best of our knowledge, this is the first work to achieve a data-driven estimate the Betti numbers of a manifold or metric-measure space, even experimentally and with the hypotheses on the probability distributions which we require.

In addition, although the present work is related to (and was directly inspired by) Diffusion Maps and related methods, it also has several important differences. In particular, the weights we choose on the graphs and simplices are significantly different than in those methods, and, in our opinion, are also more intuitive and geometric. The weights we assign to the graphs and simplicial complexes in our algorithms are simply meant to estimate the distances and $n$-dimensional volumes of the target space, whereas the weights on the graphs in Diffusion Maps and Laplacian Eigenmaps are built from estimates of a heat kernel on the graph using the known heat kernel on the ambient space. Furthermore, as mentioned above, the free parameter in our methods is a purely geometric quantity, and it is not directly related to the diffusion parameter in the induced semigroups.

 The algorithms and examples presented here give a first demonstration of the feasibility of directly estimating the real cohomology (as opposed to the persistent homology) of a metric-measure space, and open the door to the development of a number of new and exciting avenues in spectral, geometric, and topological data analysis, in addition to providing additional motivation for the further development of noncommutative methods in statistics. We include some specific questions and conjectures in Section \ref{sec:Discussion}. 
Additionally, although our present goal is only to demonstrate the estimation of the real cohomology groups $H^k(X,\R)$ in a number of examples, our estimation of the semigroups generated by the Hodge-Laplacians also has considerable intrinsic interest independent of the cohomology calculation.

Finally, while our algorithms function well in the
examples we have presented, the hypothesis that the data be uniformly
distributed is admittedly restrictive, and the methods fail when 
  the data is not close to satisfying this hypothesis. This naturally limits their 
 immediate practical use, and it will be essential
 to generalize them in future work to handle other distributions. The main
 contributions of
 our article, however, are, first, to demonstrate the existence of a class
 of algorithms which experimentally solves, in an important but very specific case, a difficult problem for which no other solution has been found, even in the restricted setting considered here, and, second, to generate concrete new conjectures
 and specific directions for future research on the questions which these methods inevitably raise.

\subsection{Organization}

In Section \ref{sec:Cohomology}, we
review the relevant algebraic topology and the combinatorial Hodge theory which we use in our algorithms. In Section \ref{sec:Estimation}, we describe the model selection procedure and introduce the
relative von Neumann entropy which is used in some of our tests. In Section \ref{sec:AlgRes}, we present the algorithms and results, and,
finally, in Section \ref{sec:Discussion}, we state a number of conjectures and
discuss avenues for future work.

\section{The Weighted Vietoris-Rips Complex and Combinatorial Hodge Theory for Metric Spaces} \label{sec:Cohomology}

We begin our discussion by reviewing some general results on finite weighted simplicial complexes and their cohomology and Hodge Laplacians, following \cite{Horak_Jost_2013}. We also introduce the specific weights which we will use for the Vietoris-Rips complexes in our algorithms. In the following, we let $\R^+ \coloneqq \{x \in \R \mid x \geq 0\}$.

\subsection{The Combinatorial Hodge Theorem and Weighted Simplicial Complexes}

We let $K$ denote an abstract simplicial complex on the set $[n] \coloneqq \{0,\dots,n-1\}$, let $S_i(K)$ denote the $i$-dimensional simplices of $K$, let $C^n(K;\R)$, $n\geq 0$ denote the simplicial cochain complexes of $K$ with real coefficients, and let $\delta_n:C^n(K;\R)\to C^{n+1}(K;\R)$ be the corresponding coboundary operator. For any choice of inner products $(\cdot,\cdot)_{C^n}$, $n\geq 0$, we define the adjoint $\delta_n^*:C^{n+1}(K;\R) \to C^{n}(K;\R)$ to be the unique operator which satisfies
\begin{equation*}
	(\delta_n v,w)_{C^{n+1}} = (v,\delta_n^*w)_{C^{n}}
\end{equation*} 
for all $v \in C^{n}$ and $w\in C^{n+1}$.

The $n$-dimensional combinatorial Hodge-Laplacian $\Delta_n(K):C^n(K;\R) \to C^n(K;\R)$ as
\begin{equation*}
	\Delta_{n}(K) := \delta_{n}^*\delta_{n} + \delta_{n-1}\delta_{n-1}^*.
\end{equation*}
We now state the combinatorial Hodge theorem due to Eckmann \cite{Eckmann_1945}.

\begin{theorem}[Eckmann 1945] For any abstract simplicial complex $K$
	with Hodge-Laplacians $\Delta_n(K)$,
	\begin{equation*}
		\ker \Delta_n(K) \cong H^k(K;\R)
	\end{equation*}
\end{theorem}

In \cite{Horak_Jost_2013}, it was shown that a choice of inner product 
$(\cdot,\cdot)_{C^n}$ is equivalent to a choice of weight function
\begin{equation*}
	w:\bigcup_{i=0}^{\dim K}S_i(K) \to \R^+
\end{equation*}
such that
\begin{align*}
	(f,g)_{C^n} &= \sum_{\sigma \in S_i(K)} w(\sigma)f(\sigma)g(\sigma)\\
	(e_\sigma,e_{\sigma'})_{C^n} &= \delta_{\sigma,\sigma'},
\end{align*}
where, for any $\sigma,\sigma' \in K$, $\delta_{\sigma,\sigma'} = 1$ if $\sigma = \sigma'$ and $0$ otherwise, and $e_\sigma:C^n(K) \to \R$ is the function
\begin{equation*}
	e_\sigma([\sigma']) = \begin{cases}
		1, & [\sigma] = [\sigma'], \\
		0, & \text{Otherwise}
	\end{cases}
\end{equation*}

\begin{definition}
	A non-negative function $w:K\to \R^+$ on a simplicial complex $K$ assigning a non-negative real number to each simplex $\sigma \in K$ is 
	called a \emph{weight function} for $K$. A pair $(K,w)$ consisting
	of a simplicial complex and a weight function is called a 
	\emph{weighted simplicial complex}.
\end{definition}

Summarizing the above discussion, any weighted simplicial
complex $(K,w)$ induces a unique collection of combinatorial
Hodge-Laplacians $\Delta_n:C^n(K;\R) \to C^n(K;\R)$. Given a weighted simplicial complex $(K,w)$ and its induced combinatorial Hodge-Laplacians $\Delta_n(K)$, $n\geq 0$, we define

\begin{definition}
	The \emph{operator (heat) semigroup generated by the Hodge-Laplacian $\Delta_n(K)$} is
	the set of operators $\{e^{-t\Delta_n}:C^n(K;\R) \to \R\mid t \geq 0\}$. Given a weighted simplicial complex $(K,w)$, we will also 
	call this the \emph{$n$-dimensional heat semigroup of $K$}
\end{definition}

We refer the reader to \cites{Jiang_etal_2011, Lim_2020, Horak_Jost_2013} for more details on
combinatorial Hodge Laplacians.

\subsection*{The Weighted Vietoris-Rips Complex}

We will require certain weights for the Vietoris-Rips complexes we
use in our algorithms. In this section, we recall the definition of the Vietoris-Rips complexes and we define these weights for use in the following. 
\begin{definition}
	Let $(X,d)$ be a semipseudometric space and $r\geq 0$ a non-negative real number. We define a simplicial complex
	$\Sigma_{VR}(X)$, called the \emph{Vietoris Rips complex of $X$ at 
		scale $r$} by
Given a metric space $(X, d_X)$ and let $r\geq 0$ be a non-negative real number. We
construct a simplicial complex $\Sigma_{VR}(X)$, called the 
\emph{Vietoris-Rips complex of $X$ at scale $r$}
in the following manner:
\begin{align*}
	S_0(\Sigma_{VR}(X)) & =  X\\
	& \vdots  \\
	S_n(\Sigma_{VR}(X)) & = \{ (x_0, \ldots, x_n) \in X \times \cdots \times X \mid
	\tmop{diam} (x_0, \ldots, x_n) \leq r \}\\
	& \vdots 
\end{align*}
where $S_n(\Sigma_{VR}(X))$ denotes the set of $n$-dimensional simplices
of $\Sigma_{VR}(X)$.
\end{definition}

To define our weights $w_{VR}$ for the Vietoris-Rips complex
of $(X,d_X)$ at scale $r\geq 0$, we start 
by assigning to each vertex weight $1$, we assign 
weights to edges by $w_{VR}({v_i,v_j}) \coloneqq d_X(v_i,v_j)$, and for 
higher-dimensional simplices, we assign the weight $w_{VR}(\{v_0,\dots,
v_n\})$ to be the volume of a Euclidean simplex with edges that have lengths
$d_X(v_i,v_j)$, where the volume is computed using Heron's formula \cite{Kock_2022}.

\begin{definition}
	Let $(X,d)$ be a metric space and let $r\geq 0$. We define the \emph{weighted
		Vietoris-Rips complex $(\Sigma_{VR}(X),w_{VR})$} at scale 
		$r$ to be the 
		pair where $\Sigma_{VR}(X)$ is the Vietoris-Rips complex 
		of $X$ at scale $r$ and $w_{VR}$ is the weight function
		defined in the paragraph above.
\end{definition}

Finally, for manifolds, the following theorem implies that the
cohomology groups of a manifold $M$ and the
Vietoris-Rips complex of a finite subset which is sufficiently Gromov-Hausdorff close to the manifold are isomorphic. 
It implies that the combinatorial Hodge-Laplacians of
such a simplicial complex may serve to model real cohomology
for the manifold.

\begin{theorem}[Latschev \cite{Latschev_2001}]
	Let $X$ be a closed Riemannian manifold. Then there exists
	$\epsilon_0 > 0$ such
	that for every $0 < \epsilon \leq \epsilon_0$
	there exists a $\delta > 0$ such that the geometric realization $|Y_\epsilon|$ of the
	Vietoris-Rips complex at scale $\epsilon$ of any metric space $Y$ which has Gromov-Hausdorff distance less than $\delta$ to $X$ is
	homotopy equivalent to $X$.
\end{theorem}

\section{Estimation of the $q$-Dimensional Heat Semigroups}\label{sec:Estimation}

We have followed three general principles when designing the procedures
presented here for estimating the topological invariants of a metric measure space $(X,d,\mu)$. First, instead of using
the point samples to construct a single space $(\hat{X}, \hat{d}, \hat{\mu})$
which is homotopy equivalent to the original space $(X, d, \mu)$, we use the
point samples $S$ to estimate the \emph{invariants} of $(X, d, \mu)$ and
nothing more. This shift in perspective adds a significant amount of
flexibility to the problem, and, in particular, it allows us to choose
different model spaces for estimating different homotopy invariants. In this
case, each dimension of the real cohomology groups of $X$ is treated separately. 
Second, we make the hypothesis that the correct model for a 
space in a given dimension is the one in which the local geometry/combinatorics reflected 
in the model is as different as
possible from its global geometry/combinatorics. This hypothesis
has been a reliable guide in constructing both the clustering algorithms in \cite{Rieser_FODS_2021} and the dimension reduction and clustering algorithms in \cite{Guzman_Rieser_2024}, and our work here is an attempt to extend
its application to higher-dimensional topological invariants. To see
why this heuristic should be expected to work, consider the 
Vietoris-Rips complexes of a finite sample of a compact manifold,
where the sample is sufficiently dense so that the Vietoris-Rips complex has the same homotopy type as the manifold for some values
of the scale parameter $r>0$, as guaranteed by Latshev's theorem. 
When the scale parameter $r>0$ of a Vietoris-Rips complex is 
too small, then the star of any given $n$-dimensional simplex will
contain much of the connected component it is contained in, and so the local and global combinatorics are expected to be similar. On the other hand, when 
the parameter $r>0$ is too large, then there is only one $n$-connected component in the simplicial complex, and much of this component will again be in the star of any given simplex, and so the local and global combinatorics of the complex will once again be expected to be close to each other. Finally, for such intermediate values of $r>0$ where the complex
correctly approximates the manifold, we expect that the star of any given simplex to be small relative to the whole (indeed, the proof of Latschev's theorem indicates exactly this), and so the local and global combinatorics diverge.  In particular, we would like
the star of any simplex in the relevant dimension to be as small as possible
(but still preserving interesting local geometry in that dimension) while still giving
the correct cohomology globally.

Guided by these principles, the idea behind our algorithms for estimating the real cohomology of a space is that
the difference between local and global geometry in dimension $q$ at scale $r
> 0$ may be measured by comparing the heat operators $e^{- s \Delta_{q,r}}$ and
$\lim_{t \to \infty} e^{- t \Delta_{q,r}}$, where $s \ll \infty$ and $\Delta_{q,r}$ is the $q$-th dimensional
Hodge Laplacian of the relevant model at scale $r$ (in this case the
Vietoris-Rips complex at scale $r$). Indeed, when $s$ is relatively small, the
spectral data of $e^{- s \Delta_{q,r}}$ contains a wealth of information about the
local $q$-dimensional combinatorics of the Vietoris-Rips complex, and when $s
\ll t$, the spectral information in $e^{- t \Delta_{q,r}}$ is mostly topological
in nature: all of the eigenvalues are either $1$ or close to $0$, and the
$1$-eigenspace is isomorphic to the $q$-dimensional real cohomology. There are
any number of ways to compare finite-dimensional operators, and in the
following, we have chosen the relative von Neumann entropy, the Hilbert-Schmidt metric, and the difference between the traces of the operators. Our estimate in any given dimension $q$ will then be given by the
scale $\hat{r}_q > 0$ for which the quantity in question is maximal.

In the remainder of this section, we recall the definition and some basic facts
about quantum relative entropy.

Our algorithm works by comparing spectral information contained in the heat semi-groups generated by the combinatorial Hodge Laplacians of 
different simplicial complexes. The spectrum of the Hodge Laplacians encode geometric information, which depends heavily on the inner product 
chosen for the cochain groups $C^i(K;\R)$, or, equivalently, the choice of weights for the simplices in $K$. The weight function which we choose 
below is meant to reflect the underlying geometry of the space approximated by the simplicial complex.

\subsection{Relative Entropy in Operator Spaces}

We say that a Hermitian operator $\rho$ on a Hilbert space $H$ is
\emph{positive} if all the eigenvalues $\lambda$ of $\rho$ satisfy
$\lambda \geq 0$, and $\rho$ is said to be \emph{strictly positive} if all
of the eigenvalues of $\rho$ are positive. We also call a positive Hermitian
operator a We define the \emph{support} of an operator $\rho$ on $H$ to be
the set
\[ \tmop{supp} \rho = \{ v \in H \mid \rho v \neq 0 \} \]
For two positive operators $\rho$ and $\sigma$ with $\tmop{supp} \rho \subset
\tmop{supp} \sigma$ we define the \emph{quantum relative entropy $H (\rho
| | \sigma)$} by
\[ H (\rho | | \sigma) \assign \tmop{Tr} (\rho \log \rho - \rho \log \sigma) .
\]
We extend the definition to the cases $\tmop{supp} \rho$ While the quantum
relative entropy is not symmetric, it is $0$ iff $\rho = \sigma$, and positive
when $\tmop{Tr} \rho$, $\tmop{Tr} \sigma \leq 1$, and so the relative entropy
provides a way to measure the degree of difference between two positive
operators. See {\cite{Wilde_2013}} for more information about quantum relative entropy.

For a one parameter semigroup $\{ e^{- t L} \}_{t \in \R}$ generated by a
finite-dimensional positive operator $L$, the next result implies that quantum
relative entropy is continuous with respect to the semigroup parameters.

\begin{theorem}
  \label{thm:Entropy diagonalizable}Let $\rho$ and $\sigma$ be simultaneously
  diagonalizable $n$-dimensional density operators such that $\tmop{supp} \rho
  \subset \tmop{supp} \sigma$. Then
  \[ H (\rho | | \sigma) = \sum_{i = 1}^n \lambda_i \log \lambda_i - \lambda_i
     {\log \lambda_i'}  \]
  where the $\lambda_i$ and $\lambda'_i$ are eigenvalues of $\rho$ and
  $\sigma$, respectively.
\end{theorem}

\begin{proof}
  If the matrices $\rho$ and $\sigma$ are simultaneously diagonalizable, then
  the expression above is the sum of the eigenvalues of $\rho \log \rho - \rho
  \log \sigma$, which is equal to the trace.
\end{proof}

We therefore conclude the following.

\begin{corollary}
  Let $L$ be a finite-dimensional positive operator, and let $s,t \in \R$. Then
  the function $h: {\R\times \R}{\rightarrow}{\R}$ given by $h (s,t) = H (e^{- s L} | |
  e^{- t L})$ is continuous.
\end{corollary}

\begin{proof}
  Follows immediately from Theorem \ref{thm:Entropy diagonalizable}.
\end{proof}

\section{Algorithms and Results}\label{sec:AlgRes}
In this section, we describe in detail our algorithms for estimating the real cohomology groups
of a metric measure space from point samples. All three methods rely on two basic
hypotheses mentioned in Section \ref{sec:Estimation}:
\begin{enumerate}
	\item That the optimal $q$-dimensional combinatorial model for a
	metric measure space in the collection of Vietoris-Rips complexes at scales $r
	> 0$ is the one in which the local and global combinatorics of the $q$-skeleta
	are maximally different, and
	\item That the difference between the local and
	global combinatorics of the $q$-skeleton of a simplicial complex may be
	measured by comparing the operators $e^{- s \Delta_q}$ and $e^{- t \Delta_q}$
	for some $0 < s \ll t$ and $0 \lll t$, where $\Delta_q$ is the $q$-dimensional Hodge Laplacian.
\end{enumerate} Combining
these assumptions with methods for measuring the difference between compact operators leads directly to our algorithms:
\begin{enumerate}
	\item First, choose a method 
	for comparing the operators $e^{-s \Delta_{q,r}}$ and $e^{-t \Delta_{q,r}}$, which we write $D(e^{-s \Delta_{q,r}},e^{-t \Delta_{q,r}})$. We assume that the method is a non-negative real-valued
	function and that a larger number indicates a larger difference between the operators, but other options are also possible. For any $r \geq 0$, write the comparison 
	\begin{equation*}
		\bar{D}_r\coloneqq \lim_{t \to \infty} D(e^{-s \Delta_{q,r}},e^{-t \Delta_{q,r}})
	\end{equation*}.
	\item Second, select a value $t_0 \gg 0$ such that $D(r)\coloneqq D(e^{-s \Delta_{q,r}},e^{-t_0 \Delta_{q,r}})$ approximates $\bar{D}_r$.
	(As we observe that this convergence is very fast in our methods, $t$ does not need to be extremely large. We chose $t_0=250$.)
	\item Third, Fix a finite value for $s$ which is much less than $t$. (We chose $1$),
	\item Next, choose $\hat{r}$ to be the value of $r$, $0<r<\text{diameter}(S)$, where $S$ is the set of point samples, such that $e^{-s \Delta_{q,r}}$ and $e^{-t_0 \Delta_{q,r}}$ are maximally different, i.e.
	\begin{equation*}
		\hat{r} \coloneqq \text{argmax } D_r
	\end{equation*}  
	\item Finally, compute the kernel of $\Delta_{q,\hat{r}}$ 
	to find the esimated cohomology group $H^q \left( X ; \R \right)$.
\end{enumerate}  
We present the pseudocode for our algorithms in Algorithm \ref{alg1} below. We considered three different options for measuring the difference between the operators $e^{-\Delta_{q,r}}$ and $e^{-t\Delta_{q,r}}$ for $t\gg0$. One is the relative von Neumann entropy, described above, a second is the Hilbert-Schmidt norm of the difference between the two matrices, and
the third is the difference of the traces of the two operators. Other quantities are of course possible as well.

\begin{algorithm}
	\caption{Estimation algorithm for the $q$-th Betti number using relative von Neumann entropy}\label{alg1}
	\begin{algorithmic}[1]
		\For{$r<Diam(S)$} 
		\State{At scale $r$, construct the graph $G_r$ and
			the Vietoris-Rips complex $\Sigma_r$ from the
			point cloud}
		\State{Compute the $q$-dimensional Hodge Laplacian $\Delta_{q,r}$, the heat operator $e^ {- \Delta_{q,r}}$ and $e^{- t_0 \Delta_r}$ for some
			$t_0$ large (we use $t_0 = 250$).}
		
		\State{Choose $s\ll t$ (we use $s=1$) and compute one of the following as $D(r)$: 
			\begin{enumerate}
				\item The relative von Neumann entropy, $D(r) \coloneqq H(\rho_r || \sigma_r)$ where $\rho =
				\frac{1}{\text{Tr }(e^{- \Delta_{q,r}})} e^{- \Delta_{q,r}}$ and $\sigma = \frac{1}{\text{Tr }(e^{-
						t_0 \Delta_{q,r}})} e^{- t_0 \Delta_{q,r}}$
				\item The Frobenius norm of the difference of the two operators $\e^{- s \Delta_{q,r}}$ and $e^{- t_0 \Delta_{q,r}}$, i.e. $D(r) \coloneqq \lVert \e^{- s \Delta_{q,r}}-e^{- t_0 \Delta_{q,r}} \rVert_{Fr}$
				\item The difference of the traces between the operators, $D(r) = Tr(e^{- s \Delta_{q,r}}-e^{- t_0 \Delta_{q,r}})$
			\end{enumerate}}
		\EndFor
		\State{Define $\hat{r} \assign \text{argmax } D(r)$.}
		\State{Compute the kernel of $\Delta_{\hat{r},q}$}
		\State{The estimated $q$-th Betti number is now $\hat{\beta}_q=\dim \ker \Delta_{q,\hat{r}}$}
	\end{algorithmic}
\end{algorithm}

We summarize the results of our experiments in Tables \ref{table:circle}-\ref{table:nonu_two_circles}, and in Figures \ref{fig:1circle}-\ref{fig:2circles_nonu} we present typical examples of the behavior of the quantities used to choose the estimators. Further experiments are forthcoming, pending the completion of ongoing improvements in our implementation. We tested the above algorithms in three experiments:
\begin{enumerate}
	\item Data Set: Fifty points sampled from a uniform distribution on a unit circle embedded in $\R^2$, corrupted with different amounts of Gaussian noise. Goal: Estimate the first Betti number (correct 
	value $= 1$). Repeated $50$ times.
	\item Data Set: Fifty points sampled from a uniform distribution on a two circles of radius $0.5$ embedded in $\R^3$. Goal: Estimate
	the first Betti number (correct value $= 2$). Repeated $50$ times.
	\item Data Set: Fifty points sampled from a non-uniform distribution on a two circles of radius $0.5$ embedded in $\R^3$. Goal: Estimate
	the second Betti number (correct value $= 2$). Repeated $10$ times.
\end{enumerate}
The algorithms were run $50$ times on the first two experiements and $10$ times in the second experiment. All three algorithms were run on the same data. As seen in the tables, the performance was very good for for both the methods using the relative von Neumann entropy and the difference of the traces as the comparison criteria, despite the small size of the data sets. The method using the Frobenius norm was not as successful.  Future experiments will include different numbers of points, more trials, and different shapes.

\begin{table}[h]
\begin{center}
\begin{tabular}{|c|c|c|c|c|c|}
\hline
\bf{Method $| \beta_{1}$} & \bf{0} & \bf{1} &  \bf{$>$1} & 
Total \# Trials & Percent correct\\ \hline

\bf{Relative Entropy} & 0 & 50 & 0 & 50 & 100\%\\ \hline
\bf{Hilbert-Schmidt Metric} & 12 & 38 & 0 & 50 & 76\%\\ \hline
\bf{Difference of Traces} & 2 & 48 & 0 & 50 & 96\% \\ \hline
\end{tabular}
\end{center}
\caption{Results of experiments using 50 sample points from a uniform distribution on a unit circle with Gaussian noise, $\sigma=0.01$. Repeated 50 times. Correct value: $\beta_1=1$} 
\label{table:circle}
\end{table}

\begin{figure}[h!]

\centering
\begin{subfigure}[]{0.4\linewidth}
\includegraphics[width=\linewidth]{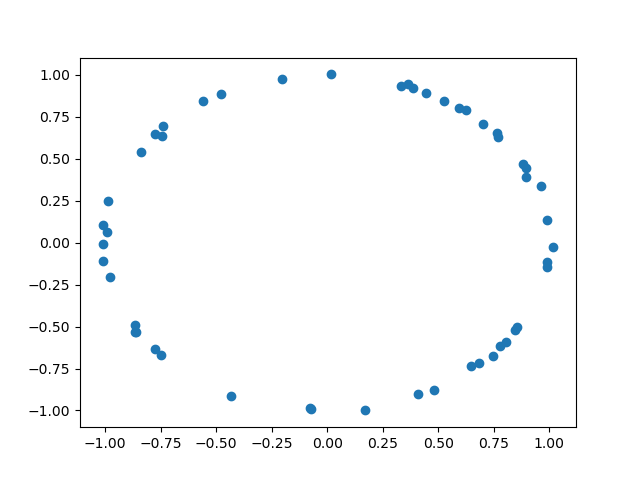}
\end{subfigure}
\begin{subfigure}[]{0.45\linewidth}
\includegraphics[width=\linewidth]{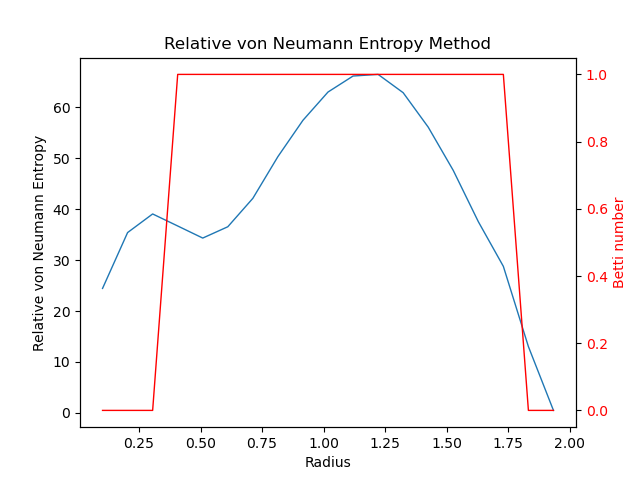}
\end{subfigure}
\begin{subfigure}[]{0.4\linewidth}
\includegraphics[width=\linewidth]{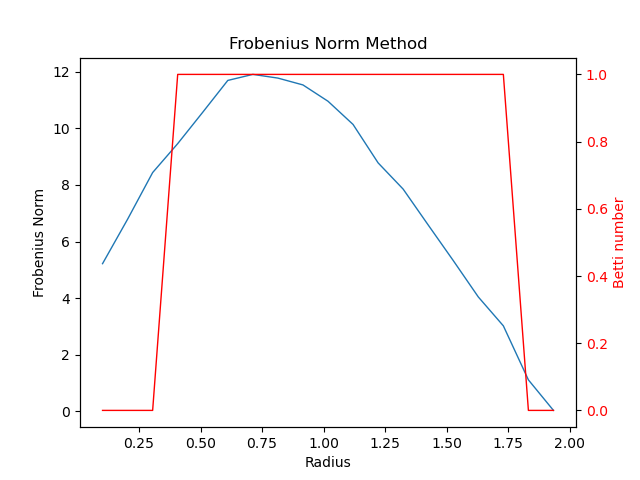}
\end{subfigure}
\begin{subfigure}[]{0.45\linewidth}
\includegraphics[width=\linewidth]{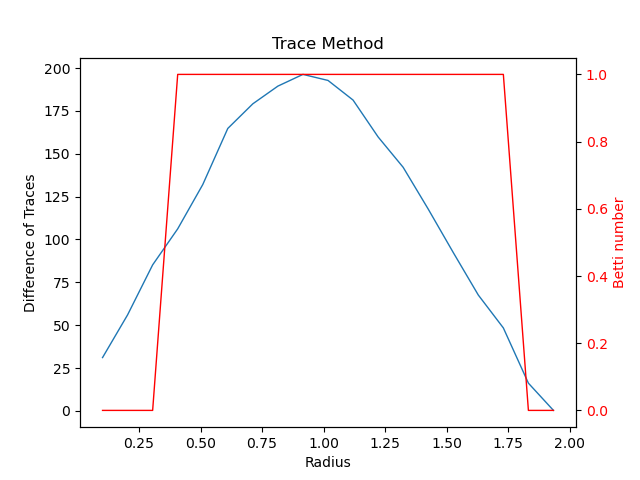}
\end{subfigure}
\caption{An example showing the typical values of the quantities used
for the combinatorial Hodge-Laplacian selection in the experiment with points uniformly sampled from a single circle.}	
\label{fig:1circle}
\end{figure}

\begin{table}[h]
\begin{center}

\begin{tabular}{|c|c|c|c|c|c|c|}
\hline
\bf{Method $| \beta_{1}$} & \bf{0} & \bf{1} & \bf{2} &  \bf{$>$2} & Total \# Trials & Percent Correct \\ \hline

\bf{Relative Entropy} & 0 & 1 & 49 & 0 & 50 & 98\% \\ \hline
\bf{Frobenius Norm} & 8 & 52 & 40 & 0 & 50 & 40\% \\ \hline
\bf{Trace} & 1 & 0 & 49 & 0 & 50 & 98\% \\ \hline
\end{tabular}
\end{center}
\caption{Results of experiments using 50 sample points from a uniform distribution on two circle with radius $=0.5$. Repeated 50 times. Correct value: $\beta_1=2$} 	
\label{table:two_circles}
\end{table}

\begin{figure}[h!]

\centering
\begin{subfigure}[]{0.4\linewidth}
\includegraphics[width=\linewidth]{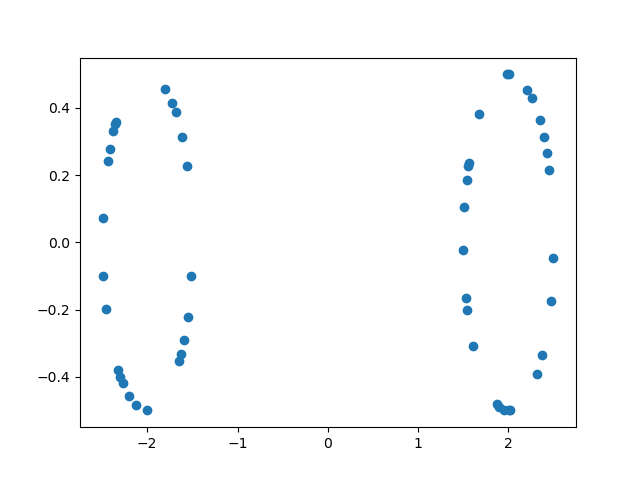}
\end{subfigure}
\begin{subfigure}[]{0.45\linewidth}
\includegraphics[width=\linewidth]{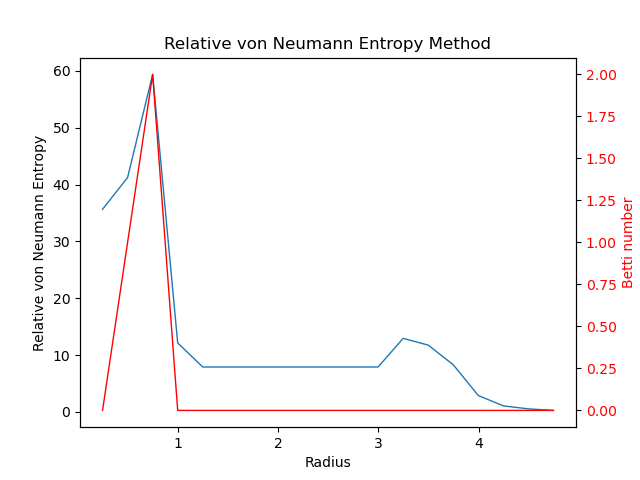}
\end{subfigure}
\begin{subfigure}[]{0.4\linewidth}
\includegraphics[width=\linewidth]{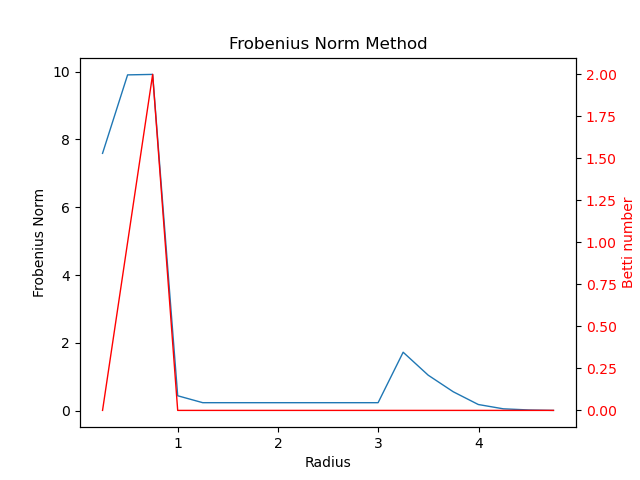}
\end{subfigure}
\begin{subfigure}[]{0.45\linewidth}
\includegraphics[width=\linewidth]{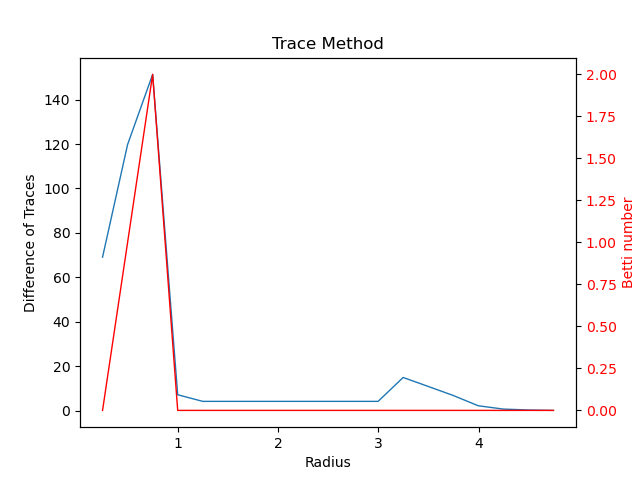}
\end{subfigure}
\caption{An example showing the typical values of the quantities used
for the combinatorial Hodge-Laplacian selection in the experiment with points uniformly sampled from a two circles embedded in 
$\R^3$. (The circles appear distorted due to the viewing angle of the image.) }	
\label{fig:2circles}
\end{figure}

\begin{table}[h]
\begin{center}

\begin{tabular}{|c|c|c|c|c|c|c|}
\hline
\bf{Method $| \beta_{1}$} & \bf{0} & \bf{1} & \bf{2} &  \bf{$>$2} & Total \# Trials & Percent Correct \\ \hline

\bf{Relative Entropy} & 0 & 23 & 27 & 0 & 50 & 54\% \\ \hline
\bf{Frobenius Norm} & 0 & 44 & 6 & 0 & 50 & 12\% \\ \hline
\bf{Trace} & 0 & 36 & 14 & 0 & 50 & 38\% \\ \hline
\end{tabular}
\end{center}
\caption{Results of the experiments using 50 sample points from a non-uniform distribution on two circles. (Correct value $\beta_1=2$). Repeated 10 times. All methods performed worse compared to the case of sampling from a uniform distribution.} 	
\label{table:nonu_two_circles}
\end{table}

\begin{figure}[h!]

\centering
\begin{subfigure}[]{0.4\linewidth}
\includegraphics[width=\linewidth]{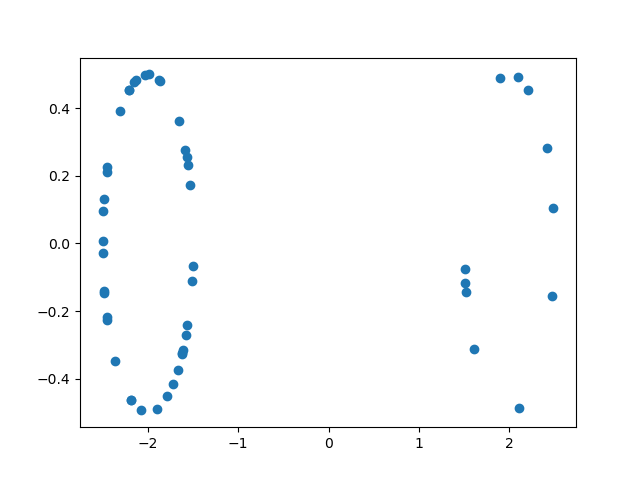}
\end{subfigure}
\begin{subfigure}[]{0.45\linewidth}
\includegraphics[width=\linewidth]{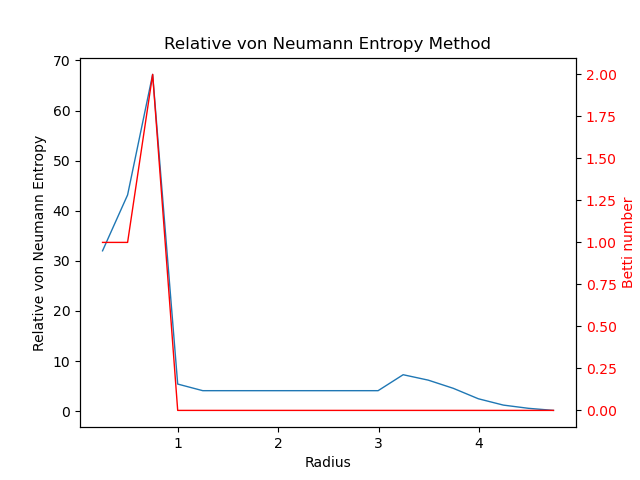}
\end{subfigure}
\begin{subfigure}[]{0.4\linewidth}
\includegraphics[width=\linewidth]{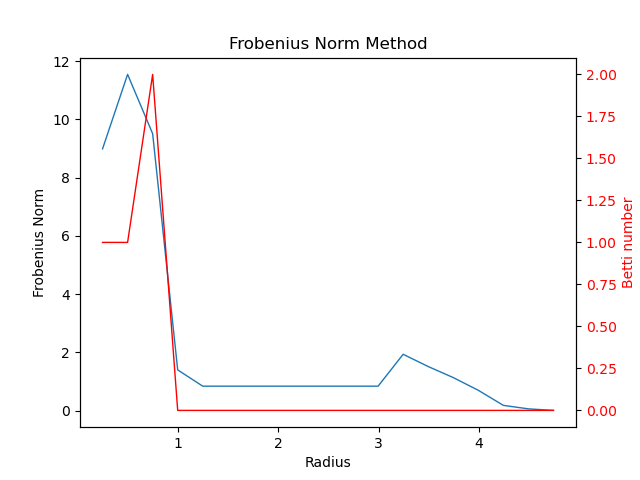}
\end{subfigure}
\begin{subfigure}[]{0.45\linewidth}
\includegraphics[width=\linewidth]{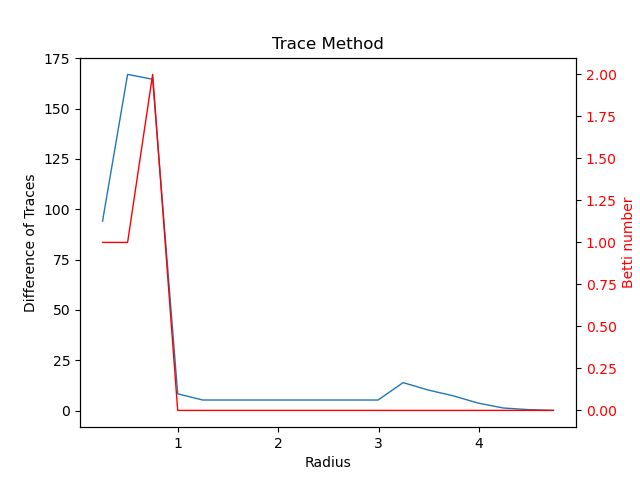}
\end{subfigure}
\caption{An example showing the typical values of the quantities used
for the combinatorial Hodge-Laplacian selection in the experiment with points non-uniformly sampled from a two circles embedded in 
$\R^3$. (The circles appear distorted due to the viewing angle of the image.) }	
\label{fig:2circles_nonu}
\end{figure}

\section{Discussion and Future Work}\label{sec:Discussion}

We have presented and compared three closely related algorithms for estimating the real cohomology groups of a 
metric-measure space from sample points, and confirmed that this procedure produces accurate results in a number of
synthetic example, sometimes corrupted with small amounts of Gaussian noise. We believe this to be the first demonstration of direct, fully data-driven estimation of the Betti number of a metric-measure space from point samples. Our methods also represent a significant departure from the usual techniques in topological data analysis, as they convert the topological problem of Betti number estimation into the analytic problem of estimating operator semigroups defined on the metric-measure spaces and then proceed to treat the problem analytically. 

The algorithms introduced here open many interesting questions about the exactly estimation of topological invariants of metric-measure
spaces from sample points, not the least of which is whether the procedures
given here stabilize with high probability to the true Betti numbers of $(X,
d, \mu)$ as the number of samples increases, as well as whether the choice of 
Hodge-Laplacians given by our procedure generates a sequences of semigroups whose spectra 
converge to the spectra of the corresponding heat semigroup on the metric-measure space as the number of points increases. We 
conjecture that both of these statements are true, at least when all of the relevant objects are defined and the sampling is from a uniform distribution, for instance when the metric-measure space is a manifold with uniform distribution. As mentioned in the Introduction, a current shortcoming of the algorithms is that they appear to be heavily sensitive to the distribution from which the points are sampled, and they cannot be expected to give sensible results if this distribution is not close to being uniform. It will
be important in future work to extend these methods to work with more general classes of distributions. These algorithms also raise important questions concerning both their non-asymptotic behavior, providing additional directions for future research. Finally, we hope that the work presented here will further motivate the development of statistical techniques in
noncommutative settings, which, among other applications, could result in new classes of methods for
the estimation of geometric and topological invariants of a space from finite sets of sample points.
\printbibliography

\end{document}